\documentclass[10pt]{article}

\usepackage[active]{srcltx}
 \usepackage{graphics,color}

\usepackage{amsmath}
\usepackage{amsfonts}
\usepackage{amssymb}
\usepackage{amscd}
\usepackage{amsthm}
\usepackage{latexsym}

\def\H{\mathcal H}

\def\P{\mathbb P}
\def\v{\mathfrak v}
\def\Tr{{\rm Tr}}
\def\d{\mathrm{d}}
\def\F{\mathcal F}
\def\S{\mathcal S}

\def\bl{\{\!\{}
\def\br{\}\!\}^B}
\def\.{ERROR}
\def\hsd{{\textstyle \frac{\hb}{2}}}
\def\rhsd{{\textstyle \frac{\sqrt{\hb}}{2}}}

 \def\X{\mathcal X} \def\C{\mathbb{C}} 
 \def\R{\mathbb{R}}

\def\L{\mathcal L}

\def\FF{{\mathfrak F}}

\def\Rep{\mathfrak{Rep}^A_{\hb}}
\def\Repp{\mathfrak{Rep}^A_{\hb'}}

\def\I{{\rm 1\kern-.26em I}}

\def\Op{\mathfrak{Op}^A_\hb}
\def\Be{\mathfrak{B}^A_{\hb}}
\def\De{\mathfrak{D}^A_{\hb}}
\def\Dep{\mathfrak{D}^A_{\hb'}}

\def\1{\mathfrak{1}}
\def\0{\mathfrak{0}}

 \def\hb{\hbar}

\newtheorem{lemma}{Lemma}[section]
\newtheorem{corollary}[lemma]{Corollary}
\newtheorem{theorem}[lemma]{Theorem}
\newtheorem{proposition}[lemma]{Proposition}
\newtheorem{definition}[lemma]{Definition}
\newtheorem{remark}[lemma]{Remark}

\numberwithin{equation}{section}

\begin{document}

\title{Positive Quantization in the Presence of a Variable Magnetic Field.}

\author{Marius M\u antoiu$\,^1$, Radu Purice$\,^2$ and Serge Richard$\,^3$}

\date{\small}
\maketitle

\begin{quote}
\begin{itemize}
\item[$^1$] Departamento de Matem\'aticas, Universidad de Chile, Las Palmeras 3425, Casilla 653,
Santiago, Chile, \\
Email: {\tt mantoiu@imar.ro}
\item[$^2$] Institute
of Mathematics Simion Stoilow of the Romanian Academy, P.O.  Box
1-764, Bucharest, \\ RO-70700, Romania,
Email: {\tt purice@imar.ro}
\item[$^3$] Graduate School of Pure and Applied Sciences,
University of Tsukuba, 1-1-1 Tennodai, Tsukuba,
Ibaraki 305-8571, Japan;
on leave from Universit\'e de Lyon; Universit\'e
Lyon 1; CNRS, UMR5208, Institut Camille Jordan, 43 blvd du 11 novembre 1918, F-69622
Villeurbanne-Cedex, France, \\
Email: {\tt richard@math.univ-lyon1.fr}
\end{itemize}
\end{quote}

\begin{abstract}
Starting with a previously constructed family of coherent states, we introduce the Berezin quantization
for a particle in a variable magnetic field and we show that it constitutes a strict quantization of a natural
Poisson algebra. The phase-space reinterpretation involves a magnetic version of the
Bargmann space and leads naturally to Berezin-Toeplitz operators.
\end{abstract}




\section*{Introduction}\label{intro}

The mathematical literature treating Berezin-Toeplitz operators in phase space
(also called anti-Wick operators, localization operators, etc)
and their connection with the pseudodifferential calculus in Weyl or Kohn-Nirenberg form is huge; we only cite some basic
references as \cite{BC1,BC2,Fo,Ha,La3,Sh}. One of the important {\it raison d'\^ etre} of this type of operators is the fact
that they realize a quantization of certain classes of physical systems, the one consisting of a spinless non-relativistic
particle being the basic one, a paradigm for the quantization of other systems. The Berezin-Toeplitz correspondence, sending
classical observables (functions on phase-space) to quantum ones (self-adjoint operators in some Hilbert space), while
inferior to the Weyl correspondence from the point of view of composition properties, has the advantage of being positive,
sending positive functions
into positive operators. It is also very often handier for norm-estimates.

It is now known that if the particle is  placed in a variable magnetic field, the Weyl form of the pseudodifferential calculus
should be modified to insure gauge covariance and to
cope with the changes in geometry and kinematics due to the presence of the magnetic field.
Recent publications \cite{IMP,IMP',KO1,LMR,MP1,MP2,MP4,MPR1,MPR2,Mu} introduced and developed a
mathematical formalism for the observables naturally associated with such a system, both in a classical and in a
quantum framework. The changes involve mathematical
objects as group $2$-cocycles with values in algebras of functions, twisted dynamical systems and twisted crossed product
$C^*$-algebras as well as an enlargement of the Weyl calculus,
so their interest is not only related to the study of physical systems in magnetic fields.
Recently, the magnetic calculus has been extended to the case of nilpotent Lie groups \cite{BB,BB3,BB4}.

Aside the quantization of observables, one must also perform the quantization of states.
A convenient
systematization of this topic is an axiomatic framework which can be found in \cite{La3}, see also
\cite{La1,La2}; it relies on seeing both the classical
and the quantum pure states as forming Poisson spaces with a transition probability.
The pure states of a classical particle are the points of the phase space $\Xi$
(the symplectic form (\ref{plectica}) takes the magnetic field into account).
On another hand, the pure states space of $\mathbb K(\H)$ (the $C^*$-algebra of all the compact operators in the Hilbert
space $\H$) is homeomorphic to the projective space $\P(\H)$. The latter space is also endowed with the $\hb$-dependent
Fubini-Study symplectic form.

Being guided by general prescriptions \cite{AAGM,Ha,KNW,La3,La4}, we defined in \cite{MPR3} a family of pure
states (called {\it magnetic coherent states}), indexed by the points of the phase space and by Planck's constant $\hb$.
They satisfy certain structural requirements and a prescribed behavior in the limit $\hb\rightarrow 0$.
We would like now to complete the picture, indicating the appropriate modifications needed to obtain {\it magnetic Berezin
operators} associated with the choice of a vector potential.
The present article outlines this topic in the setting of quantization theory, but we hope to use the
formalism in the future for concrete spectral problems involving magnetic operators.

Our first section contains a brief recall of the magnetic Weyl calculus both in pseudodifferential and in twisted
convolution form, as well as a short description of the magnetic coherent states.

In the second section the magnetic Berezin quantization is defined on functions and distributions
and its basic properties are studied. It is its fate to be
(completely) positive, but in addition it has the
important property of being {\it gauge covariant}: vector potentials corresponding to the same magnetic field lead to
unitarily equivalent Berezin operators. We study the connection with magnetic Weyl operators and show that the two quantizations
are equivalent in the limit $\hb\rightarrow 0$. The magnetic Berezin-to-Weyl map depends intrinsically on
the magnetic field $B$ and not on the choice of a corresponding potential $A$ satisfying $\d A=B$.
A very convenient setting is obtained after making a unitary transformation, which is a generalization of the classical
Bargmann transformation. The associated Bargmann-type space is a Hilbert space with reproducing kernel
and in this representation the Berezin quantization will consist of Toeplitz-type operators.
The standard Bargmann transform is build on Gaussian coherent states and this has certain advantages,
among which we quote a well-investigated holomorphic setting. The presence of the magnetic field
seems to ruin such a possibility,
so we are not going to privilege any a priori choice. It is also likely that the anti-Wick setting, involving creation and
annihilation operators and a certain type of ordering, is no longer available for variable magnetic fields.

In the last section we prove that our framework provides a strict quantization of a natural
Poisson algebra in the sense of Rieffel.
One extends in this way some of the results of \cite{BLU,Co1}, see also \cite{Co2,CX,La3}.
We prove essentially that the $\hb$-depending  magnetic Berezin operators: (i) have continuously varying operator norms,
(ii) multiply between themselves "in a classical commutative way" in the limit $\hb\rightarrow 0$, (iii) have commutators
which are governed by a magnetic Poisson bracket in the first order in $\hb$. To do this we use similar results proved in
\cite{MP2} for the magnetic Weyl calculus as well as the connection between the two magnetic quantizations that has been
obtained in the second section.
We notice that our procedure is not a {\it deformation} quantization in some obvious way, but this is not specific to the
magnetic case. For the general theory of strict quantization and for many examples we refer to \cite{La3,Ri1,Ri3,Ri4}.

\medskip
{\bf Acknowledgements:} M. M\u antoiu is partially supported by {\it N\'ucleo Cient\' ifico ICM P07-027-F
"Mathematical Theory of Quantum and Classical Magnetic Systems"} and by the Chilean Science Fundation {\it Fondecyt}
under the grant 1085162.
R. Purice was supported by CNCSIS under the Ideas Programme, PCCE Project No.
55/2008 Sisteme diferentiale in analiza neliniara si aplicatii.
S. Richard was supported by the Swiss National Science Foundation and is now supported by the Japan Society for the Promotion of Sciences.

\section{Recall of previous constructions and results}\label{geometroy}

We start by briefly reviewing the geometry of the classical system with a variable magnetic field \cite{MP2,MR},
the structure of the twisted (magnetic) calculus \cite{IMP,KO1,MP1,MPR1,Mu}
and the natural form of the magnetic coherent states \cite{MPR3}.
Details and developments are also included in the references cited above.

\subsection{The geometry of the classical system with a variable magnetic field}\label{geometry}

The particle evolves in the Euclidean space $\X:=\R^N$ under the influence of a smooth magnetic field,
which is a closed $2$-form $B$ on $\X$ ($\d B=0$), given by matrix-components
$$
B_{jk}=-B_{kj}:\X\rightarrow \R\qquad j,k=1,\dots,N.
$$
The phase space is denoted by $\Xi:=T^*\X\equiv\X\times\X^*$, where $\X^*$ the dual space of $\X$; systematic notations as
$X=(x,\xi)$, $Y=(y,\eta)$, $Z=(z,\zeta)$ will be used for its points.

The classical observables are given by real smooth functions on $\Xi$. They form a real vector space, which is also a
Poisson algebra under the usual pointwise product $(f\cdot g)(X)\equiv(fg)(X):=f(X)g(X)$ and the Poisson bracket
\begin{equation*}
\{f,g\}^B=\sum_{j=1}^N(\partial_{\xi_j}f\,\partial_{x_j}g-\partial_{\xi_j}g\,\partial_{x_j}f)+
\sum_{j,k=1}^NB_{jk}(\cdot)\,\partial_{\xi_j}f\,\partial_{\xi_k}g.
\end{equation*}
For further use, we notice that $\{\cdot,\cdot\}^B$ is canonically generated by the symplectic form
\begin{equation}\label{plectica}
(\sigma^B)_X(Y,Z)=z\cdot \eta -y\cdot\zeta+B(x)(y,z)=\sum_{j=1}^N(z_j\;\!\eta_j-y_j\;\!\zeta_j)+
\sum_{j,k=1}^NB_{jk}(x)\;\!y_j\;\!z_k,
\end{equation}
obtained by adding to the standard symplectic form
\begin{equation*}
\sigma(X,Y)\equiv\sigma[(x,\xi),(y,\eta)]:=y\cdot \xi-x\cdot\eta
\end{equation*}
a magnetic contribution.

\subsection{The structure of the magnetic pseudodifferential calculus}\label{struc}

The intrinsic way to turn to the quantum counter-part is to deform the pointwise product $fg$ into a non-commutative
product $f\sharp^B_\hb g$ depending on the magnetic field $B$ and the Planck constant $\hb$. This is given by
\begin{eqnarray}\label{composition}
\left(f\sharp^B_{\hb} g\right)(X)&:=&(\pi\hb)^{-2N}\int_\Xi
\int_\Xi \d Y \;\!\d Z\;\!e^{-\frac{2i}{\hb} \sigma(X-Y,X-Z)}\;\!
 e^{-\frac{i}{\hb}\Gamma^B\langle x-y+z,y-z+x,z-x+y\rangle}\;\!f(Y)\;\!g(Z)
\end{eqnarray}
and it involves fluxes of the magnetic field $B$ through triangles.
If $a,b,c\in \X$, then we denote by $\langle a,b,c \rangle$ the triangle in $\X$ of vertices $a,b$ and $c$ and set
$$
\Gamma^B\langle a,b,c \rangle:=\int_{\langle a,b,c\rangle}B
$$
for the invariant integration of the $2$-form $B$ through the $2$-simplex $\langle a,b,c\rangle$.
For $B=0$, (\ref{composition}) coincides with the Weyl composition of symbols in pseudodifferential theory.
Also using complex conjugation $f\mapsto\overline f$ as involution, one gets various non-commutative
$^*$-algebras of function on $\Xi$,
some of them also admitting a natural $C^*$-norm; they are regarded as {\it algebras of magnetic quantum observables}.

The full formalism also involves families of representations of these $^*$-algebras in the Hilbert space
$\H:=L^2(\X)$. They are defined by circulations $\Gamma^A[x,y]:=\int_{[x,y]}A$ of vector potentials $A$
through segments $[x,y]:=\{ty+(1-t)x\mid t\in[0,1]\}$ for any $x,y\in\X$.
We recall that, being a closed $2$-form in $\X=\R^N$, the magnetic field is exact:
it can be written as $B=\d A$ for some $1$-form $A$.
For such a vector potential $A$, we define
\begin{equation}\label{op}
\big[\Op(f)u\big](x):=(2\pi\hb)^{-N}\int_\X \int_{\X^*}\d y\;\! \d\eta \;\!
e^{\frac{i}{\hb}(x-y)\cdot\eta}\;\!
e^{-\frac{i}{\hb}\,\Gamma^A[x,y]}\;\!{\textstyle f\left(\frac{x+y}{2},\eta\right)}\;\!u(y).
\end{equation}
If $A=0$, one recognizes the Weyl quantization, associating with functions or distributions on $\Xi$ linear
operators acting on function spaces on $\X$. For suitable functions $f,g$, one proves
\begin{equation*}
\Op(f)\Op(g)=\Op(f\sharp^B_\hb g),\ \ \ \ \ \Op(f)^*=\Op(\overline f).
\end{equation*}

The main interpretation of the operators defined in (\ref{op}) is given by the formula
$\mathfrak{Op}^{A}_\hb(f)=f\left(Q;\Pi^A_\hb\right)$, where $f\left(Q;\Pi^A_\hb\right)$ should be regarded
as the function $f$ applied to the family of non-commuting self-adjoint operators
$$
\big(Q,\Pi^A_\hb\big)\equiv \big(Q_1,\dots,Q_N, \Pi^A_{\hb,1},\dots,\Pi^A_{\hb,N}\big)\ ,
$$
where $Q_j$ is the operator of multiplication by the coordinate function $x_j$ and
$\Pi^A_{\hb,j}:=-i\hb\partial_j-A_j$ is the $j$-th component of the magnetic momentum. They satisfy the commutation relations
\begin{equation}\label{rcm}
i[Q_j,Q_k]=0,\quad i[\Pi^A_{\hb,j},Q_k]=\hb\;\!\delta_{j,k},\quad i[\Pi^A_{\hb,j},\Pi^A_{\hb,k}]=-\hb\;\! B_{jk}\ .
\end{equation}
This stresses the interpretation of our twisted pseudodifferential theory  as a
{\it non-commutative functional calculus} constructed
on the commutation relations (\ref{rcm}). For $A=0$ one gets $\Pi^A_\hb=D_\hb:=-i\hb\nabla$,
so we recover the Canonical Commutation Relations of non-relativistic quantum
mechanics and the standard interpretation of the Weyl calculus.

\subsection{The twisted crossed product representation}\label{tcp}

For conceptual and computational reasons a change of realization is useful; we obtain it composing the mapping
$\Op$ with a partial Fourier transformation. Using the notation $\FF:=1\otimes \F^*$,
where $\F$ is the usual Fourier transform, we define $\Rep(F):=\Op(\FF F)$.
This makes sense for various classes of functions.
Here we record the explicit formula
\begin{equation}\label{SurRep}
[\Rep(F)u](x)= \hb^{-N}\int_\X \d y\;\!e^{-\frac{i}{\hb}\Gamma^A[x,y]}\;\!
{\textstyle F\left(\frac{x+y}{2},\frac{y-x}{\hb}\right)}\;\! u(y).
\end{equation}
Defining
\begin{equation*}
(F\diamond^B_\hb G)(x,y):=
\int_\X \d z \;\!e^{-\frac{i}{\hb}\Gamma^B\langle x-\frac{\hb}{2}y, x-\frac{\hb}{2}y + \hb z, x+\frac{\hb}{2}y \rangle}\;\!
{\textstyle F\left(x-\frac{\hb}{2}(y-z),z\right)\;\!
G\left(x+\frac{\hb}{2}z,y-z\right)}
\end{equation*}
one gets
$$
\Rep(F)\,\Rep(G)=\Rep(F\diamond^B_\hb G).
$$
In \cite{MP2,MPR1,MPR2}, the representation $\Rep$ and the composition law $\diamond^B_\hb$
have been used in connection with the $C^*$-algebraic twisted crossed product
to quantize systems with magnetic fields. We are not going to use this systematically, but only state
two basic results which are useful below: First, both the Schwartz space $\S(\X\times\X)$ and
the Banach space $L^1\big(\X_y;L^\infty(\X_x)\big)$ are stable under the multiplication $\diamond^B_\hb$.
Second, for each $F,G\in L^1\big(\X_y;L^\infty(\X_x)\big)$ one has
\begin{equation}\label{return}
\| \Rep(F)\|\,\le\,\| F\|_{1,\infty}:=\int_\Xi dy\| F(\cdot,y)\|_\infty
\end{equation}
and
\begin{equation*}
\| F\diamond^B_\hb G\|_{1,\infty}\,\le\,\| F\|_{1,\infty}\| G\|_{1,\infty}\,.
\end{equation*}

\subsection{Magnetic coherent states}\label{stric}

Let us fix a unit vector $v\in\H := L^2 (\X)$, and for any $\hb \in I:=(0,1]$
let us define the unit vector $v_\hb\in\H$ by $v_\hb(x):=\hb^{-N/4}v\big(\frac{x}{\sqrt{\hb}}\big)$.
For any choice of a vector potential $A$ generating the magnetic field $B$, we define
\emph{the family of magnetic coherent vectors associated with the pair} $(A,v)$ by
\begin{equation*}
\left[v^A_{\hb}(Z)\right](x)=e^{\frac{i}{\hb}(x-\frac{z}{2}) \cdot\zeta}
\;\!e^{\frac{i}{\hb}\Gamma^A[z,x]}\;\!v_\hb (x-z).
\end{equation*}

The pure state space of the $C^*$-algebra $\mathbb K(\H)$ of compact operators can be identified with
the projective space $\P(\H)$,
so it is natural to introduce for any $Z\in\Xi$ \emph{the coherent states} $\v^A_{\hb}(Z):\mathbb K(\H)\rightarrow\C$ by
$$
\left[\v^A_{\hb}(Z)\right](S):=
\Tr\left(\left\vert v^A_{\hb}(Z)\right> \left< v^A_{\hb}(Z) \right\vert S \right)\equiv
\left< v^A_{\hb}(Z),S\;\!v^A_{\hb}(Z)\right>.
$$

The following statement has been proved in \cite{MPR3}:
\begin{proposition}\label{ax1}
Assume that the components of the magnetic field $B$ belong to $BC^\infty(\X)$ (they are smooth and all the derivatives are
bounded) and let $v$ be an element of the Schwartz space $\S(\X)$, satisfying $\| v\|=1$.
\begin{enumerate}
\item
For any $\hb\in I$ and $u\in\H$ with $\|u\|=1$, one has
\begin{equation}\label{par}
\int_{\Xi}\frac{\d Y}{(2\pi\hb)^N}\;\big|\langle v^A_{\hb}(Y),\,u\rangle\big|^2=1.
\end{equation}
\item
For any $Y,\,Z\in \Xi$, one has
\begin{equation*}
\lim_{\hb\rightarrow 0}\big|\langle v^A_{\hb} (Z),\,v^A_{\hb} (Y)\rangle\big|^2=\delta_{ZY}.
\end{equation*}
\item
If $g:\Xi\rightarrow\C$ is a bounded continuous function and $Z\in \Xi$, one has
\begin{equation}\label{richard}
\lim_{\hb\rightarrow 0}\int_\Xi\frac{\d Y}{(2\pi\hb)^N}\;\!\big|\langle v^A_{\hb}(Z),
v^A_{\hb}(Y)\rangle\big|^2 \,g(Y)\,=\,g(Z) .
\end{equation}
Furthermore, if $g \in \S(\Xi)$ then one has
\begin{equation}\label{ultimita}
\underset{\hb\rightarrow 0}{\lim}\left[\v^A_{\hb} (Z)\right]\left[\mathfrak{Op}^A_\hb(g)\right]=\delta_{Z}(g)=g(Z)\ .
\end{equation}
\end{enumerate}
\end{proposition}

\section{Magnetic Berezin operators}\label{erezin}

Although not always necessary, for the sake of uniformity, we shall always assume that $v \in \S(\X)$ and that
$B_{jk}\in BC^\infty(\X)$ for $j,k\in \{1,\dots,N\}$. The components of the corresponding vector potentials $A$ will always belong to
$C^\infty_{\rm{pol}}(\X)$, {\it i.e.}~they are smooth and all the derivatives are polynomially bounded.
This can obviously be achieved under our assumption on $B$, and this will facilitate
subsequent computations involving the Schwartz class.

\subsection{The magnetic Berezin quantization}\label{subiectul}

The following is an adaptation of \cite[Def.~II.1.3.4]{La3}:

\begin{definition}
\emph{The magnetic Berezin quantization associated with the set of coherent states
$\big\{\v^A_\hb (Z)\mid Z \in \Xi, \hb \in I\big\}$}
is the family of linear mappings $\big\{\Be:L^{\infty} (\Xi)\rightarrow \mathbb B(\H)\big\}_{\hb\in I} $ given for any
$f \in L^\infty(\Xi)$ by
\begin{equation}\label{obs}
\Be(f):=\int_{\Xi}\frac{\d Y}{(2\pi\hb)^N}\;\!f(Y)\;\!\v^A_{\hb}(Y)\ ,
\end{equation}
where $\v^A_{\hb}(Y)$ is seen as the rank one projection $|v^A_\hb(Y)\rangle \langle v^A_\hb(Y)|$.
\end{definition}

Note that for any unit vector $u \in \H$ and for the corresponding element $\mathfrak u \in \P(\H)$ one has
$$
{\mathfrak u}\big( \Be(f)\big)=\Tr \big(|u\rangle \langle u | \Be(f)\big)=\big\langle u,\Be (f)u\big\rangle=
\int_{\Xi}\frac{\d Y}{(2\pi\hb)^N}f(Y)\big|\langle v^A_{\hb}(Y),u\rangle\big|^2\ .
$$
By (\ref{par}), this offers a rigorous interpretation of
\eqref{obs} as a weak integral; it can be regarded as a Bochner integral only under an integrability condition on $f$.
The explicit function $H^A_{\hb,u}(\cdot):=(2\pi\hb)^{-N}\big|\langle v^A_{\hb}(\cdot),u\rangle\big|^2$
deserves to be called {\it the magnetic Husimi function associated to the vector $u$} \cite{Ha,La3}. It is a positive phase space
probability distribution.

\begin{proposition}\label{zzz}
The following properties of the Berezin quantization hold:
\begin{enumerate}
\item $\Be$ is a linear map satisfying $\| \Be (f)\| \leq \| f\|_\infty$, $\forall f\in L^\infty(\Xi)$\,.
\item $\Be$ is positive, {\it i.e.}~for any $f\in L^\infty (\Xi)$ with $f\geq 0$ {\it a.e.}~one has $\Be (f) \geq 0$\,.
\item If $f\in L^1 (\Xi)\cap L^{\infty}(\Xi)$, then $\Be (f)$ is a trace-class operator and
$$
\Tr\big[\Be (f)\big]=\int_\Xi \frac{\d Y}{(2\pi\hb)^N}\;\!f(Y) \,.
$$
\item For any $g\in L^1(\Xi)$ one has
$ \int_\Xi \d Y\;\!\v^A_{\hb}(Y)\big(\Be(g)\big) =\int_\Xi \d Z\,g(Z)$\,.
\item
Let us denote by $C_0(\Xi)$ the $C^*$-algebra of all complex continuous functions on $\Xi$ vanishing at infinity.
Then $\Be\left[C_0 (\Xi)\right]\subset\mathbb K(\mathcal H)$.
\end{enumerate}
\end{proposition}

\begin{proof}
Most of the properties are quite straightforward, and they are true in a more abstract setting \cite[Thm.~II.1.3.5]{La3}.
The fourth statement is a simple consequence of (\ref{par}).
By the point 3 one has $\Be[C_{\rm{c}}(\Xi)]\subset\mathbb
K(\mathcal H)$; we denoted by $C_{\rm{c}}(\Xi)$ the space of continuous compactly supported functions on $\Xi$.
This, the point 1 and the density of $C_{\rm{c}}(\Xi)$ in $C_0(\Xi)$
imply that $\Be[C_{0}(\Xi)]\subset\mathbb K(\mathcal H)$.
\end{proof}

\begin{remark}
{\rm To extend the weak definition of $\Be(f)$ to distributions, remark that one can write
\begin{equation*}
\left<u_1,\Be(f)u_2\right>=\int_\Xi \frac{\d Y}{(2\pi\hb)^N}\;\!f(Y)\left[w^A_\hb(u_1,v)\right](Y)
\overline{\left[w^A_\hb(u_2,v)\right](Y)},
\end{equation*}
where
\begin{equation*}
\left[w^A_\hb(u,v)\right](Y):=\left<u,v^A_\hb(Y)\right>=\int_\X \d x\,e^{\frac{i}{\hb}\left(x-y/2\right)\cdot\eta}\,
e^{\frac{i}{\hb}\Gamma^A[y,x]}\,\overline{u(x)}\,v_\hb(x-y)\,.
\end{equation*}
A simple computation shows that $w^A_\hb(u,v)$ is obtained from $\overline u\otimes v$ by applying successively a linear change
of variables, multiplication with a function belonging to $C^\infty_{{\rm pol}}(\X\times\X)$ and a partial Fourier transform.
All these operations are isomorphisms between the corresponding Schwartz spaces, so $w^A_\hb(u,v)\in\S(\Xi)$ if $u,v\in\S(\X)$,
and the mapping $(u,v)\mapsto w^A_\hb(u,v)$ is continuous.
It follows that $w^A_\hb(u_1,v)\,\overline{w^A_\hb(u_2,v)}\in\S(\Xi)$, so one can define for $f\in\S'(\Xi)$
the linear continuous operator $\Be(f):\S(\X)\rightarrow \S'(\X)$ by
\begin{equation*}
\left<u_1,\Be(f)u_2\right>=(2\pi\hb)^{-N}
\left<\overline{w^A_\hb(u_1,v)}\,w^A_\hb(u_2,v),f\right>\,,
\end{equation*}
using in the r.h.s. the duality between $\S(\Xi)$ and $\S'(\Xi)$.}
\end{remark}

An important property that should be shared by any quantization procedure in the presence of a magnetic field is
{\it gauge covariance}. Two vector potentials $A$ and $A'$ which differ only by the differential
$\d\rho$ of a $1$-form (function) will clearly generate the same magnetic field. It is already known \cite{MP1}
that the magnetic Weyl operators $\mathfrak{Op}^A_\hb(f)$ and $\mathfrak{Op}^{A'}_\hb(f)$ are unitarily equivalent.
The next result expresses the gauge covariance of the magnetic Berezin quantization.

\begin{proposition}\label{gc}
If $A'=A+\d\rho$, then $\mathfrak B_\hb^{A'}(f)=e^{\frac{i}{\hb}\rho(Q)}\Be(f)e^{-\frac{i}{\hb}\rho(Q)}$.
\end{proposition}

\begin{proof}
A simple computation gives $v^{A'}_{\hb}(Y)=e^{-\frac{i}{\hb}\rho(y)}e^{\frac{i}{\hb}\rho(Q)}v^A_{\hb}(Y)$ for every $Y\in\Xi$.
Then it follows that
$\v^{A'}_{\hb}(Y)=e^{\frac{i}{\hb}\rho(Q)}\v^A_{\hb}(Y)e^{-\frac{i}{\hb}\rho(Q)}$
and this implies the result.
\end{proof}

\medskip
\noindent
{\bf Some particular cases:}
\begin{enumerate}
\item
Clearly we have for all $Z\in\Xi$
\begin{equation*}
\Be(\delta_Z)=(2\pi\hb)^{-N}|v^A_\hb(Z)\rangle\langle v^A_\hb(Z)|\,,
\end{equation*}
so the coherent states (seen as rank one projections) are magnetic Berezin operators in a very explicit way.
Notice that $\Be(\delta_Z)$ is a compact operator although $\delta_Z$ does not belong to $L^\infty(\Xi)$.

\item
For $f:=\varphi\otimes 1$, with $\varphi:\X\rightarrow\mathbb C$ (polynomially bounded), a simple computation leads to
\begin{equation*}
\left \langle u,\Be(f)u\right \rangle=
\int_\X\int_\X \d x\,\d y\,\varphi(x-\sqrt\hb y)\,|u(x)|^2\,|v(y)|^2\,.
\end{equation*}
Setting $\varphi(x):=x_j \equiv q_j(x,\xi)$ for some $j\in\{1,\dots,N\}$, one gets
\begin{equation*}
\left \langle u,\Be(f)u\right \rangle=\int_\X \d x\,x_j\,|u(x)|^2 -\sqrt{\hb}\,\| u\|^2\int_\X \d y\,y_j\,|v(y)|^2.
\end{equation*}
Thus, if $v$ is even, then $\Be(q_j)=Q_j$. In general we only get this in the limit $\hb\rightarrow 0$.

\item
If we set $f(x,\xi):=\xi_j\equiv p_j(x,\xi)$ then
\begin{eqnarray*}
\left\langle u,\Be(p_j)u\right\rangle
&=&\int_\X\int_\X \d x\,\d y\,\partial_{x_j}\left\{\Gamma^A[x,x-\sqrt\hb y]\right\}|u(x)|^2\,|v(y)|^2
+ i\sqrt{\hb} \, \|u\|^2 \int_\X \d y [\partial_j v](y)\overline{v(y)} \\
&&+ i\hb\int_\X \d x\,\overline{[\partial_j u](x)}\,u(x)
\,.
\end{eqnarray*}
\end{enumerate}

\subsection{Connection with the Weyl quantization}\label{weyl}

A natural question is to find the magnetic Weyl symbol of a Berezin operator.
For computational reasons a change of realization is useful; we obtain it by composing the mapping $\mathfrak B^A_\hb$
with a partial Fourier transformation. Using again the notation $\FF:=1\otimes \F^*$,
where $\F$ is the usual Fourier transform, we define $\mathfrak D^A_\hb(F):=\Be(\FF F)$.
This makes sense for various classes of functions, but we are only going to use them for $F\in\S(\X\times\X)$.
Here we only record the explicit formula
\begin{equation}\label{obsM}
\big[\De(F)u\big](x)=\hb^{-N}\int_\X\int_\X\d y \;\!\d z
\;\!{\textstyle F\left(z,\frac{y-x}{\hb}\right)} \;\!v_\hb(x-z)\;\!\overline{v_\hb(y-z)}
\;\!e^{-\frac{i}{\hb}\Gamma^A[x,z]}\;\!e^{-\frac{i}{\hb}\Gamma^A[z,y]}\;\!u(y)\,.
\end{equation}

It is easier to prove first:

\begin{proposition}\label{sfirsacheM}
For any $F \in \S(\X\times \X)$ one has
$\De(F)=\Rep\big[\Sigma^B_\hb(F)\big]$ with $\Sigma^B_\hb(F)$ given by
\begin{equation}\label{surK}
\big[\Sigma^B_\hb(F)\big](x,y):=\int_\X\d z \;\!
F(x-\sqrt{\hb}\;\!z,y)\;\!
{\textstyle
\overline{v\left(z+\frac{\sqrt{\hb}}{2}y\right)}\;\!v\left(z-\frac{\sqrt{\hb}}{2}y\right)}
\;\!e^{-\frac{i}{\hb}\Gamma^B\langle x-\sqrt{\hb}\;\!z,x+\frac{\hb}{2}y,x-\frac{\hb}{2}y \rangle}
\,.
\end{equation}
The mapping $\Sigma^B_\hb$ extends to a linear contraction of the Banach space $L^1\big(\X_y;C_0(\X_x)\big)$.
\end{proposition}

\begin{proof}
By comparing \eqref{obsM} with \eqref{SurRep} and using Stokes' Theorem to write the sum of three circulations of
$A$ as the flux of $B=\d A$ through the corresponding triangle, one gets:
\begin{eqnarray*}
\big[\Sigma^B_\hb(F)\big]{\textstyle \left(\frac{x+y}{2},\frac{y-x}{\hb}\right)}
&=& \int_\X \;\!\d z
\;\!{\textstyle F\left(z,\frac{y-x}{\hb}\right)} \;\!v_\hb(x-z)\;\!\overline{v_\hb(y-z)}
\;\!e^{-\frac{i}{\hb}\Gamma^A_\hb[x,z]}\;\!e^{-\frac{i}{\hb}\Gamma^A_\hb[z,y]}\;\!e^{-\frac{i}{\hb}\Gamma^A_\hb[y,x]} \\
&=&\int_\X \;\!\d z
\;\!{\textstyle F\left(z,\frac{y-x}{\hb}\right)} \;\!v_\hb(x-z)\;\!\overline{v_\hb(y-z)}
\;\!e^{-\frac{i}{\hb}\Gamma^B\langle z,y,x\rangle}\ .
\end{eqnarray*}
Then, some simple changes of variables lead to the above expression.

For any $x,y\in\X$ one clearly has
\begin{equation}\label{baba}
\big|\big[\Sigma^B_\hb(F)\big](x,y)\big|
\leq  \|F(\cdot,y)\|_{\infty} \;\! \|v\|_{2}^2=\|F(\cdot,y)\|_{\infty}\,,
\end{equation}
so $\Sigma^B_\hb$ is a contraction if on $\S(\X\times\X)$ we consider the norm of $L^1\big(\X_y;L^\infty(\X_x)\big)$.
The function $x\mapsto\left[\Sigma_\hb^B(F)\right](x,y)$ is continuous and vanishes as $x \to \infty$,
 by an easy application of the Dominated Convergence Theorem.
We conclude that $\Sigma^B_\hb[\S(\X\times\X)]\subset L^1\big(\X_y;C_0(\X_x)\big)$.
Also using (\ref{baba}) and the density of $\S(\X\times\X)$ in $L^1\big(\X_y;C_0(\X_x)\big)$,
this completes the proof of the statement.
\end{proof}

\begin{remark}\label{gradina}
{\rm With some extra work, one could show that $\Sigma_\hb^B$ sends continuously $\S(\X\times \X)$ into $\S(\X\times \X)$.
This relies on the assumption $v\in\S(\X)$ and uses
polynomial estimates on the magnetic phase factor in (\ref{surK}), which can be extracted rather easily from the
fact that all the derivatives of $B$ are bounded. We shall not use this result.}
\end{remark}

\begin{corollary}\label{larry}
For any $f \in \S(\Xi)$ one has $\mathfrak B^A_{\hb}(f)=\Op\big[\mathfrak S^B_\hb(f)\big]$, with
\begin{equation*}
\big[\mathfrak S^B_\hb(f)\big](x,\xi):=\int_\X\int_{\X^*}\d z \;\!\d\zeta \;\!
f (x-z,\xi-\zeta)\;\!\Upsilon^B_\hb(x;z,\zeta)
\end{equation*}
and
$$
\Upsilon^B_\hb(x;z,\zeta):=(2\pi)^{-N}\int_{\X}\d y\;\!e^{-iy\cdot\zeta}\;\!{\textstyle
\overline{v_\hb\left(z+\frac{\hb}{2}y\right)}\;\!v_\hb\left(z-\frac{\hb}{2}y\right)} \;\!
e^{-\frac{i}{\hb}\Gamma^B\langle x-z, x+\frac{\hb}{2}y,x-\frac{\hb}{2}y\rangle}
\ .
$$
\end{corollary}

\begin{proof}
Our previous definitions imply that $\mathfrak S^B_\hb(f)=\mathfrak F\left[\left(\Sigma^B_\hb(\mathfrak F^{-1}f)\right)\right]$.
The Corollary follows from Proposition \ref{sfirsacheM} by a straightforward computation.
\end{proof}

\begin{remark}\label{mamarc}
{\rm It is satisfactory that  $\mathfrak S^B_\hb(f)$ depends only on $B$ and not on the vector potential $A$.
If $B=0$ (or if it is constant), then $\Upsilon^B_\hb$ does not depend on $x$ and the operation $\mathfrak S_\hb^0$
is just a convolution, as expected.}
\end{remark}

Let us turn now to the study of the $\hb\rightarrow 0$ behavior of the magnetic Berezin quantization. We record first the
following simple consequence of (\ref{richard}) and (\ref{ultimita}):

\begin{proposition}\label{iamboss}
For any $X\in\Xi$  and any bounded continuous function $g:\Xi \to \C$, one has
\begin{equation*}
\lim_{\hb\rightarrow 0}\left \langle v^A_{\hb}(X),\left[\mathfrak B^A_{\hb}(g)\right]v^A_\hb(X)\right \rangle=g(X)\ .
\end{equation*}
Furthermore, if $g\in \S(\Xi)$ then
\begin{equation*}
\lim_{\hb\rightarrow 0}\big \langle v^A_{\hb}(X),\big[\mathfrak{Op}^A_\hb(g)-\mathfrak B^A_{\hb}(g)\big]v^A_{\hb}(X)\big \rangle=0.
\end{equation*}
\end{proposition}

Next we would like to show that the representation $\Op$ and the Berezin quantization
are equivalent in the limit $\hb\rightarrow 0$, thus improving on the second statement of the previous Proposition.
We start with a result that will be used below and that might have some interest in its own. For that purpose, let us set $\overline I:=\{0\}\cup I=[0,1]$ and $\Sigma^B_0:={\rm id}$.

\begin{proposition}\label{own}
The map $\overline I\ni\hb\rightarrow \Sigma^B_\hb\in\mathbb B\left[L^1\big(\X_y;C_0(\X_x)\big)\right]$ is strongly continuous.
\end{proposition}

\begin{proof} We are going to check that for any $F\in\S(\X\times\X)$, one has
\begin{equation}\label{stigmatt}
\|\Sigma^B_\hb(F)-F\|_{1,\infty}\rightarrow 0\ \ \ {\rm when}\ \ \hb\rightarrow 0.
\end{equation}
By the density of $\S(\X\times\X)$
in $L^1\big(\X_y;C_0(\X_x)\big)$ this will prove the continuity in $\hb=0$, which is the most interesting result.
Continuity in other values $\hb\in I$ is shown analogously and is left as an exercise.

Let us first observe that
\begin{eqnarray*}
&&\big[\Sigma^B_\hb(F)-F\big](x,y)\\
&=& \int_\X\d z \;\!F(x-\sqrt{\hb}\;\!z,y)\;\!
{\textstyle
\overline{v\left(z+\frac{\sqrt{\hb}}{2}y\right)}\;\!v\left(z-\frac{\sqrt{\hb}}{2}y\right)}
\;\!e^{-\frac{i}{\hb}\Gamma^B\langle x-\sqrt{\hb}\;\!z,x+\frac{\hb}{2}y,x-\frac{\hb}{2}y \rangle}
-F(x,y)\\
&=& \int_\X\d z \Big[F(x-\sqrt{\hb}\;\!z,y)\;\!
{\textstyle
\overline{v\left(z+\frac{\sqrt{\hb}}{2}y\right)}\;\!v\left(z-\frac{\sqrt{\hb}}{2}y\right)}
\;\!e^{-\frac{i}{\hb}\Gamma^B\langle x-\sqrt{\hb}\;\!z,x+\frac{\hb}{2}y,x-\frac{\hb}{2}y \rangle}
-F(x,y)\;\!\overline{v(z)}\;\!v(z)\Big]\\
&=:&\int_\X \d z \;\!J_\hb(x,y;z)\ .
\end{eqnarray*}
Furthermore, one clearly has
\begin{equation*}
\big|\big[\Sigma^B_\hb(F)-F\big](x,y)\big|
\,\leq 2 \;\!\|F(\cdot,y)\|_{\infty} \;\! \|v\|_2^2\,=\,2\|F(\cdot,y)\|_{\infty}\,,
\end{equation*}
with the r.h.s.~independent of $x$ and which belongs to $L^1(\X_y)$. It then follows from the Dominated Convergence
Theorem that \eqref{stigmatt} holds if for each $y \in \X$ one has
\begin{equation*}
\lim_{\hb \to 0} \sup_{x\in \X} \big|\big[\Sigma^B_\hb(F)-F\big](x,y)\big| = 0\ .
\end{equation*}

For that purpose, let $r \in \R_+$ and set $B_r$ for the ball centered at $0 \in \X$ and of radius $r$ and $B_r^\bot$
for the complement $\X\setminus B_r$. Observe then that for any fixed $y\in \X$ and for $r$ large enough one has
\begin{eqnarray}\label{1terme}
\nonumber \big|\big[\Sigma^B_\hb(F)-F\big](x,y)\big|
&\leq& \int_{B_r} \d z \;\!|J_\hb(x,y;z)|
+ \int_{B_r^\bot} \d z \;\!|J_\hb(x,y;z)| \\
&\leq& \int_{B_r} \d z \;\!|J_\hb(x,y;z)| +
\|F(\cdot,y)\|_{\infty} \;\! \|v\|_{\infty} \;\!\left(
\|v\|_{L^1(B_{r-|y|/2}^\perp)}+\|v\|_{L^1(B_r^\bot)}\right)\ .
\end{eqnarray}
Clearly, the second term of \eqref{1terme} is independent of $x$ and $\hb$ and can be made arbitrarily small by
choosing $r$ large enough.
The hypothesis $F \in \S(\X\times \X)$ implies that
$\int_{B_r} \d z \;\!|J_\hb(x,y;z)|$ can also be made arbitrarily small (independently of $\hb\in I$) by restricting $x$
to the complement of a large compact subset of $\X$. Since $F$, $v$ and the magnetic phase factor are all continuous,
for $x$ and $z$ restricted to compact subsets of $\X$
the integrant $J_\hb(x,y;z)$ can be made arbitrarily small by choosing $\hb$ small enough. Thus the first
term of \eqref{1terme} also has a vanishing limit as $\hb \to 0$.
\end{proof}

\begin{corollary}\label{whoknowsM}
For any $f\in\S(\Xi)$ one has
\begin{equation*}
\lim_{\hb \to 0}\big\|\Be(f)-\Op(f)\big\|=0\ .
\end{equation*}
\end{corollary}

\begin{proof}
By using the notations above, Proposition \ref{sfirsacheM} and the fact that $\FF$ is an isomorphism
from $\S(\Xi)$ to $\S(\X \times \X)$, one has to show for any $F\in\S(\X\times\X)$ that
$$
\lim_{\hb \to 0}\big\|\Rep\big[\Sigma^B_\hb(F)-F\big]\big\|=0\ .
$$
However, by (\ref{return}), this follows from (\ref{stigmatt}).
\end{proof}

\subsection{Operators in the Bargmann representation}\label{Ttz}

We now introduce the generalization to our framework of the Bargmann transform and consider the fate of the Berezin operators
in the emerging realization.
The proofs of the statements bellow are straightforward; most of them are not specific to our magnetic framework, see \cite[Sec.~II.1.5]{La3}.

\begin{definition}
\begin{enumerate}
\item
The mapping
$\mathcal U^A_{\hb} : L^2(\X)\rightarrow L^2_\hb (\Xi)\equiv L^2\left(\Xi;\frac{\d X}{(2\pi\hb)^N}\right)$ given by
$\left(\mathcal U^A_{\hb} u\right)(X):=\left< v^A_{\hb} (X), u\right>$
is called {\rm the Bargmann transformation} corresponding to the family of coherent vectors $\left(v^A_{\hb}
(X)\right)_{X\in\Xi}$.
\item
The subspace $\mathcal K^A_{\hb} :=\mathcal U^A_{\hb} \big(L^2 (\X)\big)\subset L^2_\hb (\Xi)$ is called
{\rm the magnetic Bargmann space} corresponding to the family of coherent vectors $\left(v^A_{\hb} (X)\right)_{X\in \Xi}$.
\end{enumerate}
\end{definition}

First we remark that $\mathcal U^A_{\hb}$ is an isometry with adjoint
$$
\left(\mathcal U^A_{\hb}\right)^* : L^2_\hb (\Xi)\rightarrow L^2(\X),\ \ \left(\mathcal U^A_{\hb}\right)^* \Phi
:=\int_{\Xi}\frac{\d X}{(2\pi\hb)^N}\Phi (X) v^A_{\hb} (X)
$$
and final projection $P^A_{\hb}:=\mathcal U^A_{\hb}\left(\mathcal U^A_{\hb}\right)^*\in\mathbb P [L^2_\hb (\Xi)]$,
with $P^A_{\hb}\big(L^2_\hb (\Xi)\big)=\mathcal K^A_{\hb}$. The integral kernel of this projection
$$
K^A_{\hb} :\Xi\times\Xi\rightarrow\mathbb C,\ \ K^A_{\hb}(Y,Z):=\left< v^A_{\hb} (Y),  v^A_{\hb} (Z)\right>\,,
$$
explicitly equal to
$$
K^A_{\hb}(Y,Z)=e^{\frac{i}{2\hb}(y\cdot\eta-z\cdot\zeta)}\int_\X \d x\,e^{\frac{i}{\hb}x\cdot(\zeta-\eta)}\,
e^{-\frac{i}{\hb}\Gamma^A[y,x]}\,e^{\frac{i}{\hb}\Gamma^A[z,x]}
\,\overline{v_\hb(x-y)}\,v_\hb(x-z)\ ,
$$
is a continuous function and it is a reproducing kernel for $\mathcal K^A_{\hb}$:
$$
\Phi (Y)=\int_{\Xi}\frac{\d Z}{(2\pi\hb)^N} K^A_{\hb} (Y,Z)\Phi (Z),\
\ \forall\,Y\in\Xi,\ \ \forall\,\Phi\in\mathcal K^A_{\hb}.
$$
The magnetic Bargmann space is composed of continuous functions and contains all the vectors $K^A_{\hb} (X,\cdot),\,X\in\Xi$.
The evaluation maps $\mathcal K^A_{\hb}\ni\Phi\rightarrow\Phi(X)\in\mathbb C$ are all continuous.
Furthermore, the set of vectors $\Psi^A_{\hb} (X):=\mathcal U^A_{\hb} \big(v^A_{\hb} (X)\big)$ with $X\in\Xi$ forms a family of coherent states
in the magnetic Bargmann space.

\begin{proposition}\label{plitz}
For $f\in L^\infty (\Xi)$ the operator
$$
\mathfrak T^A_{\hb} (f):=\mathcal U^A_{\hb}\mathfrak B^A_{\hb} (f)
\big(\mathcal U^A_{\hb}\big)^*\equiv \int_{\Xi}\frac{\d Y}{(2\pi\hb)^N}f(Y)|\Psi^A_\hb(Y)\rangle\langle\Psi^A_\hb(Y)|
$$
takes the form of a Toeplitz operator $\mathfrak T^A_{\hb} (f)=P^A_{\hb}M_f P^A_{\hb}$,
where $M_f$ is the operator of multiplication by $f$ in the Hilbert space $L^2_\hb (\Xi)$.
\end{proposition}

\begin{proof}
Simple computation.
\end{proof}

Now, let us denote by ${\left< \cdot\,,\,\cdot \right>} _{(\hb)}$ the
scalar product of the space $L^2_\hb (\Xi)$.

\begin{definition}
{\rm The magnetic covariant symbol} of the operator $S\in\mathbb B
[L^2_\hb (\Xi)]$ is the function
$$
s^A_{\hb} (S):\Xi\rightarrow\mathbb C,\ \ \big[s^A_{\hb} (S)\big](X):=\left<\Psi^A_{\hb} (X),S\Psi^A_{\hb}
(X) \right>_{(\hb)}.
$$
\end{definition}

Of course, this can also be written as
$$
\big[s^A_{\hb} (S)\big](X)=\left< v^A_{\hb} (X), \big(\mathcal U^A_{\hb}\big)^*\,S\mathcal U^A_{\hb} v^A_{\hb} (X)\right>,
$$
which suggests the definition of the magnetic covariant symbol of an operator $T\in \mathbb B[L^2 (\X)]$ to be
$$
\big[t^A_{\hb} (T)\big](X)=\left< v^A_{\hb} (X), T\,v^A_{\hb} (X)\right>=\left[\v^A_{\hb}(X)\right](T).
$$

Sometimes $\mathfrak B^A_{\hb} (f)$ and $\mathfrak T^A_{\hb} (f)$ are called operators with {\it contravariant symbol} $f$.
We avoided the Wick/anti-Wick terminology, since its full significance involving ordering is not clear here.

\begin{remark}
{\rm
A straightforward calculation leads to the covariant symbol of a Toeplitz operator
 $$
\left(s^A_{\hb} \left[\mathfrak T^A_{\hb}
(f)\right]\right)(X)=\int_{\Xi}\frac{\d Y}{(2\pi\hb)^N}f(Y)\,\left\vert\left< v^A_{\hb} (X),v^A_{\hb} (Y) \right>
\right\vert^2.
 $$
The relation (\ref{richard}) shows that {\it the magnetic Berezin transformation}, sending a continuous and bounded function $f$ on $\Xi$ to
$s^A_{\hb} \big[\mathfrak T^A_{\hb} (f)\big]$,
converges to the identity operator when $\hb\rightarrow 0$.
}
\end{remark}

\section{Strict quantization}\label{ueyl}

This section is dedicated to a proof of

\begin{theorem}\label{remaiat}
Assume that $v \in \S(\X)$, that
$B_{jk}\in BC^\infty(\X)$ for $j,k\in \{1,\dots,N\}$ and that a corresponding vector potentials $A$ with components in  $C^\infty_{\rm{pol}}(\X)$ has been chosen. Then, the magnetic Berezin quantization $\mathfrak B^A_{\hb}$ is a strict quantization
of the Poisson algebra $\left(\S(\Xi;\mathbb R),\cdot,\{\cdot,\cdot\}^B\right)$.
\end{theorem}
In technical terms this means that
\begin{enumerate}
\item
For any real $f\in \S(\Xi)$, the map $I \ni \hb \mapsto \big\|\Be(f)\big\|\in \R_+$ is continuous and it extends continuously
to $\overline I:=[0,1]$ if the value $\| f\|_\infty$ is assigned to $\hb=0$ ({\it Rieffel's axiom}).
\item
For any $f,g \in \S(\Xi)$ the following property holds ({\it von Neumann's axiom}):
\begin{equation*}
\lim_{\hb \to 0}\big\|
{\textstyle \frac{1}{2}} \big[
\Be(f)\;\!\Be(g) + \Be(g)\;\!\Be(f)\big] - \Be(f g)
\big\|=0\ .
\end{equation*}
\item
For any $f,g \in \S(\Xi)$, the following property holds ({\it Dirac's axiom}):
\begin{equation*}
\lim_{\hb \to 0} \big\|\textstyle{\frac{1}{i\hb}}
\big[\Be(f), \Be(g)\big]-\Be\big(\{f,g\}^B\big)\big\|=0\ .
\end{equation*}
\end{enumerate}

Equivalently, we intend to show that the map $\De$ defines a strict quantization of the Poisson algebra
$\big(\S(\X\times\X;\R),\diamond^B_0,\bl\cdot,\cdot\br\big)$, where the product $\diamond^B_0$ and
the Poisson bracket $\bl\cdot,\cdot\br$ are deduced from the Poisson algebra $\big(\S(\Xi;\R),\cdot,\{\cdot,\cdot\}^B\big)$
through the partial Fourier transformation $\FF$.
One obtains easily that
\begin{equation*}
\left(F\diamond^B_0G\right)(x,y)={\textstyle \frac{1}{(\sqrt{2\pi})^N}}
\int_\X \d z\;\!
{\textstyle F\left(x,z\right)\;\!
G\left(x,y-z\right)}\ .
\end{equation*}
Similarly, the Poisson bracket is given by
\begin{equation*}
\bl F,G\br = (2\pi)^{N/2}\Big[
\sum_j \big[(Y_j F) \diamond^B_0 \big({\textstyle \frac{1}{i}}\;\!\partial_{x_j} G\big) -
\big({\textstyle \frac{1}{i}}\;\!\partial_{x_j}x F\big)\diamond^B_0 (Y_j G)\big] - \sum_{j,k} B_{jk}
\;\!(Y_j F)\diamond^B_0 (Y_k G)\Big]
\end{equation*}
with $[Y_j F](x,y) = y_j F(x,y)$ and $[\partial_{x_j} F](x,y)=\frac{\partial F}{\partial x_j}(x,y)$.

Our approach relies on a similar proof \cite{MP2} for the fact that $\Rep$ defines a strict quantization
of the Poisson algebra $\big(\S(\X\times\X;\R),\diamond^B_0,\bl\cdot,\cdot\br\big)$. This and the results of subsection
\ref{weyl} will lead easily to the first two conditions. Dirac's axiom is more difficult to check;
it relies on some detailed calculations and estimates.

\subsection{Rieffel's condition}

The most important information is of course
\begin{equation*}
\underset{\hb\rightarrow 0}{\lim}\| \Be(f)\|\,=\,\| f\|_\infty.
\end{equation*}
This follows easily from the analog relation proved in \cite{MP2} for the magnetic Weyl quantization $\Op$ and from Corollary
\ref{whoknowsM}. For convenience, we treat also continuity outside $\hb=0$.

\begin{proposition}
For any $F\in \S(\X\times \X)$, the map $I \ni \hb \mapsto \big\|\De(F)\big\|\in \R_+$ is continuous.
\end{proposition}

\begin{proof}
We first recall that it has been proved in \cite{MP2} that the map
$I \ni \hb \mapsto \big\|\Rep(F)\big\| \in \R_+$ is continuous for any $F \in \S(\X\times\X)$.
Let $\hb,\hb' \in I$ and $F\in \S(\X\times \X)$. Then one has
\begin{eqnarray}\label{faitfroid}
\nonumber \big\|\De(F)-\Dep(F)\big\| &=& \big\|\Rep\big[\Sigma^B_\hb(F)\big]-\Repp\big[\Sigma^B_{\hb'}(F)\big]\big\| \\
&\leq& \big\|\Rep\big[\Sigma^B_\hb(F)-\Sigma^B_{\hb'}(F)\big]\big\| +
\big\|\big(\Rep-\Repp\big)\big[\Sigma^B_{\hb'}(F)\big]\big\|\ .
\end{eqnarray}
Since the inequality $\|\Rep(G)\| \leq \|G\|_{1,\infty}$ always holds, the first term of
\eqref{faitfroid} goes to $0$ as $\hb' \to \hb$ by Proposition \ref{own}.
The second term also vanishes as $\hb'\to\hb$ by the result of \cite{MP2} recalled above and by a simple
approximation argument. The statement then easily follows.
\end{proof}

\subsection{Von Neumann's condition}

One has to show that for any $F,G \in \S(\X \times \X)$ the following property holds:
\begin{equation*}
\lim_{\hb \to 0}\big\|
{\textstyle \frac{1}{2}} \big[
\De(F)\;\!\De(G) + \De(G)\;\!\De(F)\big] - \De(F\diamond_0^B G)
\big\|=0\ .
\end{equation*}
In fact, since $F\diamond^B_0 G = G\diamond^B_0 F$, it is enough to show that
\begin{equation*}
\lim_{\hb \to 0}\big\|
\De(F)\;\!\De(G)  - \De(F\diamond_0^B G) \big\|=0\ .
\end{equation*}
By taking the previous results into account, one has
\begin{eqnarray*}
\big\|\De(F)\;\!\De(G)- \De(F\diamond_0^B G)\big\| &=&
\big\|\Rep\big[\Sigma^B_\hb(F)\big]\;\!\Rep\big[\Sigma^B_\hb(G)\big]
-\Rep\big[\Sigma^B_\hb(F\diamond^B_0 G)\big]\big\| \\
&=& \big\|\Rep\big[\Sigma^B_\hb(F)\diamond^B_\hb \Sigma^B_\hb(G)-
\Sigma^B_\hb(F\diamond^B_0 G)\big]\big\| \\
&\leq& \big\|\Sigma^B_\hb(F)\diamond^B_\hb \Sigma^B_\hb(G)-
\Sigma^B_\hb(F\diamond^B_0 G)\big\|_{1,\infty}\\
&\leq& \big\|\Sigma^B_\hb(F) - F\big\|_{1,\infty} \;\!
\big\| \Sigma^B_\hb(G)\big\|_{1,\infty} +
\big\|\Sigma^B_\hb(G) - G\big\|_{1,\infty} \;\!
\| F\|_{1,\infty} \\
&&+\big\|\Sigma^B_\hb\big(F\diamond^B_0 G\big)-
F\diamond^B_0 G\big\|_{1,\infty} +
\big\|F\diamond^B_\hb G - F\diamond^B_0 G\big\|_{1,\infty}\ .
\end{eqnarray*}
It has been shown in Proposition \ref{own} that
$\big\|\Sigma^B_\hb(H) - H\big\|_{1,\infty}$ converges to $0$ as $\hb\to 0$ for any $H \in \S(\X\times \X)$.
Using this and the fact that $\Sigma^B_\hb$ is a contraction in $L^1\big(\X_y;C_0(\X_x)\big)$,
it follows that the first three terms above vanish as $\hb$ goes to $0$.
Finally, the convergence of $\big\|F\diamond^B_\hb G - F\diamond^B_0 G\big\|_{1,\infty}$ to $0$ as $\hb \to 0$
has been proved in \cite{MP2} in a more general context.

\subsection{Dirac's condition}

One has to show that for any $F,G \in \S(\X\times\X)$, the following result holds:
\begin{equation*}
\lim_{\hb \to 0} \big\|\textstyle{\frac{1}{i\hb}}
\big[\De(F), \De(G)\big]-\De\big(\bl F,G\br\big)\big\|=0,
\end{equation*}
which is equivalent to
\begin{equation*}
\lim_{\hb \to 0} \Big\|\textstyle{\frac{1}{i\hb}}\;\!
\Rep\Big(\Sigma^B_\hb(F) \diamond^B_\hb \Sigma^B_\hb(G) - \Sigma^B_\hb(G)\diamond^B_\hb \Sigma^B_\hb(F)\Big)-
\Rep\Big(\Sigma^B_\hb\big( \bl F,G \br \big)\Big\|=0 \ .
\end{equation*}
By taking into account the previous results, this reduces to showing that
\begin{equation*}
\lim_{\hb \to 0} \big\|\textstyle{\frac{1}{i\hb}}\;\!
\big(\Sigma^B_\hb(F) \diamond^B_\hb \Sigma^B_\hb(G) - \Sigma^B_\hb(G)\diamond^B_\hb \Sigma^B_\hb(F)\big)
-\bl F,G\br \big\|_{1,\infty}=0 \ .
\end{equation*}

For simplicity, let us denote by $V(a,b,c,d)$ the product $\overline{v(a)}v(b)\overline{v(c)}v(d)$
and let $\Gamma^B(a,b,c,d,e)$ be the flux of the magnetic field through the "pentagon" of vertices $a,b,c,d,e$.
With these notations one has
\begin{eqnarray*}
&&\big[\Sigma^B_\hb(F)\diamond^B_\hb \Sigma^B_\hb(G)\big](x,y) \\
&=&(2\pi)^{-N/2} \int_\X\int_\X\int_\X \d z \;\!\d a\;\! \d b\;\!F\big(x-\sqrt{\hb}a - \hsd (y-z),z\big)\;\!
G\big(x-\sqrt{\hb}b+\hsd z, y-z\big)\\
&&\cdot \ V\big(a+\rhsd z,a-\rhsd z,b+\rhsd (y-z),b -\rhsd (y-z) \big) \\
&&\cdot \
e^{-\frac{i}{\hb}\Gamma^B(x-\frac{\hb}{2} y,x-\sqrt{\hb}a - \frac{\hb}{2} (y-z), x-\frac{\hb}{2} y + \hb z,
x-\sqrt{\hb}b+\frac{\hb}{2} z,x+\frac{\hb}{2} y)} \ .
\end{eqnarray*}
Then, with some simple changes of variables it follows that
\begin{eqnarray*}
&&\big[\Sigma^B_\hb(F)\diamond^B_\hb \Sigma^B_\hb(G)-
\Sigma^B_\hb(G)\diamond^B_\hb \Sigma^B_\hb(F)\big](x,y) \\
&=&(2\pi)^{-N/2}\;\!\int_\X\int_\X\int_\X \d z \;\!\d a\;\! \d b\;\!  V\big(a+\rhsd z,a-\rhsd z,b+\rhsd (y-z),b -
\rhsd (y-z) \big) \\
&&\cdot \ \Big [
F\big(x-\sqrt{\hb}a - \hsd (y-z),z\big)\;\!G\big(x-\sqrt{\hb}b+\hsd z, y-z\big)\;\!
w_1^B(x,y,z,a,b;\hb)\\
&& \ - F\big(x-\sqrt{\hb}a+\hsd (y-z), z\big)\;\!G\big(x-\sqrt{\hb}b - \hsd z,y-z\big) \;\!
w_2^B(x,y,z,a,b;\hb) \Big]\ .
\end{eqnarray*}
with
$$
w_1^B(x,y,z,a,b;\hb):= e^{-\frac{i}{\hb}\Gamma^B(x-\frac{\hb}{2} y,x-\sqrt{\hb}a - \frac{\hb}{2} (y-z),
x-\frac{\hb}{2} y + \hb z,x-\sqrt{\hb}b+\frac{\hb}{2} z,x+\frac{\hb}{2} y)}
$$
and
$$
w_2^B(x,y,z,a,b;\hb):=e^{-\frac{i}{\hb}\Gamma^B(x-\frac{\hb}{2} y,x-\sqrt{\hb}b - \frac{\hb}{2} z, x+\frac{\hb}{2}
y - \hb z,x-\sqrt{\hb}a+\frac{\hb}{2} (y-z),x+\frac{\hb}{2} y)}\ .
$$
By using the Taylor development for $\varepsilon$ near $0$ :
\begin{eqnarray*}
F(x+\varepsilon \;\!y,z)&=& F(x,z)+\varepsilon \sum_j y_j \int_0^1 \d s \;\![\partial^{x_j} F](x+s\;\!\varepsilon\;\! y,z) \\
&=:&F(x,z)+ \L\big(F;x,\varepsilon y,z\big)\ ,
\end{eqnarray*}
the term between square brackets can be rewritten as the sum of the following four terms:
\begin{equation*}
I_1(x,y,z,a,b;\hb):=F\big(x-\sqrt{\hb}a,z\big)\;\!G\big(x-\sqrt{\hb}b, y-z\big)\;\!
\big[w_1^B(x,y,z,a,b;\hb)-w_2^B(x,y,z,a,b;\hb)\big]\ ,
\end{equation*}
\begin{eqnarray*}
&&I_2(x,y,z,a,b;\hb):=F(x-\sqrt{\hb}a,z)\cdot \\
&&\ \cdot \Big[ \L\big(G;x-\sqrt{\hb}b,\hsd z,y-z\big)\;\!
w_1^B(x,y,z,a,b;\hb) - \L\big(G;x-\sqrt{\hb}b, - \hsd z,y-z\big) \;\! w_2^B(x,y,z,a,b;\hb) \Big]\ ,
\end{eqnarray*}
\begin{eqnarray*}
&&I_3(x,y,z,a,b;\hb):=G(x-\sqrt{\hb}b,y-z)\cdot \\
&&\ \cdot \Big[\L\big(F;x-\sqrt{\hb}a, - \hsd (y-z),z\big)\;\!
w_1^B(x,y,z,a,b;\hb) - \L\big(F;x-\sqrt{\hb}a,\hsd (y-z),z\big)\;\!
w_2^B(x,y,z,a,b;\hb)\Big]
\end{eqnarray*}
and
\begin{eqnarray*}
I_4(x,y,z,a,b;\hb)&:=&
\L\big(F;x-\sqrt{\hb}a, - \hsd (y-z),z\big)\;\!\L\big(G;x-\sqrt{\hb}b,\hsd z,y-z\big) \;\! w_1^B(x,y,z,a,b;\hb) \\
&& - \L\big(F;x-\sqrt{\hb}a,\hsd (y-z),z\big)\;\!\L\big(G;x-\sqrt{\hb}b, - \hsd z,y-z\big)\;\!
w_2^B(x,y,z,a,b;\hb)\ .
\end{eqnarray*}

The term $I_1$ is going to be studied below. Then, observe that $I_2(x,y,z,a,b;\hb) + I_3(x,y,z,a,b;\hb)$ is equal to
\begin{eqnarray*}
&&\hsd \;\!F(x-\sqrt{\hb}a,z) \sum_j z_j \int_0^1\d s\Big[ [\partial_{x_j} G]\big(x-\sqrt{\hb}b+\hsd s z,y-z\big)
\;\!w_1^B(x,y,z,a,b;\hb) \\
&&\qquad + [\partial_{x_j} G]\big(x-\sqrt{\hb}b - \hsd sz,y-z\big) \;\! w_2^B(x,y,z,a,b;\hb)\Big] \\
&&- \hsd \;\! G(x-\sqrt{\hb}b,y-z)\sum_j (y_j-z_j) \int_0^1\d s\Big[ [\partial_{x_j} F]\big(x-\sqrt{\hb}a -
\hsd s(y-z),z\big)\;\!w_1^B(x,y,z,a,b;\hb)\\
&&\qquad +[\partial_{x_j} F]\big(x-\sqrt{\hb}a+\hsd s(y-z),z\big)\;\!
w_2^B(x,y,z,a,b;\hb)\Big]\ .
\end{eqnarray*}
Furthermore, the term $I_4(x,y,z,a,b;\hb)$ clearly belongs to $O(\hb^2)$, for fixed $x,y,z,a$ and $b$.
So, let us now concentrate on the main part of $I_1$ :

\begin{lemma}
For fixed $x,y,z,a$ and $b$ one has
\begin{equation}\label{ouf}
\lim_{\hb \to 0} {\textstyle \frac{1}{i\hb}}\big[w_1^B(x,y,z,a,b;\hb)-w_2^B(x,y,z,a,b;\hb)\big] =
-\sum_{j,k}z_j\;\!(y_k-z_k)\;\!B_{jk}(x)\ .
\end{equation}
\end{lemma}

\begin{proof}
Since $|w_j^B|=1$ one has $w_1^B-w_2^B=w_1^B(1-(w_1^B)^{-1}\;\!w_2^B)$. Furthermore, one has
\begin{eqnarray*}
&& w_1^B(x,y,z,a,b;\hb)^{-1}\;\!w_2^B(x,y,z,a,b;\hb) \\
&=&e^{-\frac{i}{\hb}\Gamma^B\langle x+\frac{\hb}{2}y,x-\frac{\hb}{2}y+\hb z,x-\frac{\hb}{2}y\rangle}\;\!
e^{-\frac{i}{\hb}\Gamma^B\langle
x-\frac{\hb}{2}y,x+\frac{\hb}{2}y-\hb z,x+\frac{\hb}{2}y\rangle} \ \cdot \\
&& \cdot \ e^{-\frac{i}{\hb}\Gamma^B\langle x-\frac{\hb}{2}y+\hb z ,x-\sqrt{\hb}a -\frac{\hb}{2}(y-z),
x-\frac{\hb}{2}y\rangle}\;\! e^{-\frac{i}{\hb}\Gamma^B\langle
x+\frac{\hb}{2}y-\hb z,x-\sqrt{\hb}a+\frac{\hb}{2}(y-z) ,x+\frac{\hb}{2}y\rangle} \ \cdot \\
&& \cdot \ e^{-\frac{i}{\hb}\Gamma^B\langle x+\frac{\hb}{2}y ,x-\sqrt{\hb}b +\frac{\hb}{2}z,
x-\frac{\hb}{2}y+\hb z\rangle}\;\! e^{-\frac{i}{\hb}\Gamma^B\langle
x-\frac{\hb}{2}y,x-\sqrt{\hb}b-\frac{\hb}{2}z ,x+\frac{\hb}{2}y-\hb z\rangle}\\
&=:&[L^B_1\cdot L^B_2\cdot L^B_3](x,y,z,a,b;\hb)\ .
\end{eqnarray*}
By using the standard parametrization of the flux through triangles, one then obtains
\begin{eqnarray*}
L^B_1(x,y,z,a,b;\hb)&:=&\exp\Big\{-i\hb\sum_{j,k}(y_j-z_j)\;\!z_k\int_0^1\d \mu \int_0^1\d \nu\;\!\mu \ \cdot \\
&&\quad \cdot \ \Big[B_{jk}\big(x+\hsd y-\mu \hb (y-z)-\mu\nu \hb z\big)
+B_{jk}\big(x-\hsd y+\mu \hb (y-z)-\mu\nu \hb z\big)
\Big]\Big\}\ ,
\end{eqnarray*}
\begin{eqnarray*}
&&L^B_2(x,y,z,a,b;\hb):=\exp\Big\{i\sum_{j,k}\big(a_j+\rhsd z_j\big) \;\!\big(a_k-\rhsd z_k\big)\int_0^1\d \mu
\int_0^1\d \nu\;\!\mu \ \cdot \\
&&\quad \cdot \ \Big[B_{jk}\big(x-\mu \sqrt{\hb} a(1-\nu) -\hsd\big[ y-z(2 -\mu -\mu\nu)\big]\big)
-B_{jk}\big(x-\mu\sqrt{\hb} a (1-\nu) + \hsd \big[y- z(2-\mu - \mu\nu)\big]\big)
\Big]\Big\}
\end{eqnarray*}
and
\begin{eqnarray*}
&&L^B_3(x,y,z,a,b;\hb):=\exp\Big\{i\sum_{j,k}\big(b_j+\rhsd (y_j-z_j\big) \;\!\big(b_k-\rhsd (y_k-z_k)\big)
\int_0^1\d \mu \int_0^1\d \nu\;\!\mu \ \cdot \\
&& \quad\cdot \ \Big[B_{jk}\big(x  -\mu\sqrt{\hb}b(1-\nu) +\hsd \big[y - (y-z) \mu(1 +\nu)\big]\big) -
B_{jk}\big(x-\mu\sqrt{\hb}b  (1-\nu)  -\hsd\big[ y -(y-z) \mu (1 +\nu)\big] \big)
\Big]\Big\} \ .
\end{eqnarray*}

Now, let us observe that
\begin{eqnarray*}
{\textstyle \frac{1}{i\hb}}[w^B_1-w^B_2]&=& (w^B_1)^{-1}\;\!{\textstyle \frac{1}{i\hb}}[1-L^B_1\;\!L_2^B\;\!L^B_3] \\
&=& (w^B_1)^{-1}\;\!{\textstyle \frac{1}{i\hb}}[1-L^B_1] + (w^B_1)^{-1}\;\!L^B_1 \;\!{\textstyle \frac{1}{i\hb}}[1-L^B_2]+
(w^B_1)^{-1}\;\!L^B_1 \;\!L^B_2\;\!{\textstyle \frac{1}{i\hb}}[1-L^B_3]\ .
\end{eqnarray*}
By taking the limit $\hb \to 0$ and by taking the equality $B_{jk}=-B_{kj}$ into account,
the first term leads to the r.h.s.~of \eqref{ouf}. For the other two terms, by a Taylor development of
the magnetic field one easily obtains that their limit as $\hb \to 0$ is null.
\end{proof}

By adding these different results, one can now prove:
\begin{proposition}[Dirac's condition]
For any $F,G \in \S(\X\times\X)$, the following property holds:
\begin{equation*}
\lim_{\hb \to 0} \big\|\textstyle{\frac{1}{i\hb}}
\big[\De(F), \De(G)\big]-\De\big(\bl F,G\br\big)\big\|=0\ .
\end{equation*}
\end{proposition}

\begin{proof}
By considering the results obtained above, the proof simply consists in numerous applications of the
Dominated Convergence Theorem and in various approximations as in Proposition \ref{own}.
The normalization $\|v\|_{L^2(\X)}=1$ should also been taken into account.
\end{proof}



\begin{thebibliography}{00}

\bibitem{AAGM} S.T. Ali, J-P. Antoine, J-P. Gazeau and U.A. M\"uller: {\it Coherent States and Their Generalizations:
A Mathematical Overview}, Rev. Math. Phys. {\bf 7}, 1013--1104, (1990).

\bibitem{BB} I. Beltita and D. Beltita: {\it Magnetic Pseudo-differential Weyl Calculus on Nilpotent Lie Groups,} Ann. Global
Anal. Geom. {\bf 36}, 293--322, (2009).

\bibitem{BB3} I. Beltita and D. Beltita: {\it Modulation Spaces of Symbols for Representations of Nilpotent Lie
Groups}, to appear in J. Fourier Anal. Appl.

\bibitem{BB4} I. Beltita and D. Beltita: {\it Continuity of Magnetic Weyl Calculus}, Preprint ArXiV.

\bibitem{BC1} C.~A. Berger and L.~A. Coburn: {\it Toeplitz Operators on the Segal Bargmann Space,}
Trans. Amer. Math. Soc. {\bf 301}, 813--829, (1987).

\bibitem{BC2} C.~A. Berger and L.~A. Coburn: {\it Heat Flow and Berezin Toeplitz Estimates,}
Amer. J. Math. {\bf 116}, 563--590, (1994).

\bibitem{BLU} D. Borthwick, A. Lesniewski and H. Upmeier: {\it Non-perturbative Deformation Quantization of Cartan Domains},
J. Funct. Anal. {\bf 113}, 153--176, (1993).

\bibitem{Co1} L.~A. Coburn: {\it Deformation Estimates for the Berezin Toeplitz Quantization,}
Comm. Math. Phys. {\bf 149}, 415--424, (1992).

\bibitem{Co2} L.~A. Coburn: {\it The Measure Algebra of the Heisenberg Group,} J. Funct. Anal. {\bf 161}, 509--525, (1999).

\bibitem{CX} L.~A. Coburn and J. Xia: {\it Toeplitz Algebras and Rieffel Deformations,} Comm. Math. Phys. {\bf 168},
23--38, (1995).

\bibitem{Fo} G.~B. Folland: {\it Harmonic Analysis in Phase Space}, Princeton University Press, Princeton, New Jersey, 1989.

\bibitem{Ha} B. C. Hall: {\it Holomorphic Methods in Analysis and Mathematical Physics}, Contemp. Math.
{\bf 260}, 1--59, (2000).

\bibitem{IMP} V. Iftimie, M. M\u antoiu and R. Purice: {\it Magnetic Pseudodifferential Operators},
Publ. RIMS. {\bf 43}, 585--623, (2007).

\bibitem{IMP'} V. Iftimie, M. M\u antoiu and R. Purice: {\it A Beals-Type Criterion for Magnetic Pseudodifferential
Operators}, Comm. in PDE {\bf 35}, 1058--1094, (2010).

\bibitem{KNW} D. Kaschek, N. Neumaier and S. Waldmann: {\it Complete Positivity of Rieffel's Deformation Quantization},
J. Noncommut. Geom. {\bf 3}, 361--375, (2009).

\bibitem{KO1} M.V. Karasev and T.A. Osborn: {\it Symplectic Areas,
Quantization and Dynamics in Electromagnetic Fields}, J. Math.
Phys. {\bf 43}, 756--788, (2002).

\bibitem{La1} N.~P. Landsman: {\it Classical Behaviour in Quantum Mechanics: A Transition Probability Approach},
Int. J. of Mod. Phys. {\bf B}, 1545--1554, (1996).

\bibitem{La2} N.~P. Landsman: {\it Poisson Spaces with a Transition Probability},
Rev. Math. Phys. {\bf 9}, 29--57, (1997).

\bibitem{La3} N.~P. Landsman: {\it Mathematical Topics Between Classical and Quantum Mechanics}, Springer-Verlag,
New-York, 1998.

\bibitem{La4} N.~P. Landsman: {\it Quantum Mechanics on Phase Space}, Stud. Hist. Phil. Mod. Phys. {\bf 30}, 287--305, (1999).

\bibitem{LMR} M. Lein, M. M\u antoiu and S. Richard: {\it Magnetic
Pseudodifferential Operators with Coefficients in $C^*$-algebras}, to appear in Publ. RIMS.

\bibitem{MP1} M. M\u antoiu and R. Purice: {\it The Magnetic
Weyl Calculus}, J. Math. Phys. {\bf 45}, 1394--1417 (2004).

\bibitem{MP2} M. M\u antoiu and R. Purice: {\it Strict Deformation Quantization for a Particle in a Magnetic Field},
J. Math. Phys. {\bf 46}, 052105, (2005).

\bibitem{MP4} M. M\u antoiu and R. Purice: {\it The Modulation Mapping for Magnetic Symbols and Operators},
Proc. of the AMS {\bf 138}, 2839--2852, (2010).

\bibitem{MPR1} M. M\u antoiu, R. Purice and S. Richard: {\it Twisted
Crossed Products and Magnetic Pseudodifferential Operators},
in Operator Algebras and Mathematical Physics, 137--172, Theta Ser. Adv. Math. 5, Theta, Bucharest, 2005.

\bibitem{MPR2} M. M\u antoiu, R. Purice and S. Richard: {\it Spectral and Propagation Results
for Magnetic Schr\"odinger Operators; a $C^*$-Algebraic Approach},
J. Funct. Anal. {\bf 250}, 42--67, (2007).

\bibitem{MPR3} M. M\u antoiu, R. Purice and S. Richard: {\it Coherent States in the Presence of a Variable Magnetic Field},
to appear in Int. J. of Geom. Meth. in Mod. Phys.

\bibitem{MR} J. Marsden and T. Ratiu: {\it Introduction to Mechanics and Symmetry}, Springer-Verlag, Berlin, New-York, 1994.

\bibitem{Mu} M. M\"uller: {\it Product Rule for Gauge Invariant Weyl Symbols and its Application to the
Semiclassical Description of Guiding Center Motion}, J. Phys. A:
Math. Gen. {\bf 32}, 1035--1052, (1999).

\bibitem{Ri1} M. Rieffel: {\it Deformation Quantization of
Heisenberg Manifolds}, Comm. Math. Phys. {\bf 122}, 531--562, (1989).

\bibitem{Ri3} M. Rieffel: {\it Deformation Quantization for Actions of $\mathbb R^d$}, Memoirs of the AMS {\bf 106}, 1993.

\bibitem{Ri4} M. Rieffel: {\it Quantization and $C^*$-Algebras},
In: Doran R.S. (ed) {\it $C^*$-Algebras 1943-1993}. {\it Contemp.
Math.} {\bf 167}, 67--97, A.M.S. Providence, 1994.

\bibitem{Sh} M. Shubin: {\it Pseudodifferential Operators and Spectral Theory}, Springer Series in Soviet Math., Springer, 1987.

\end{thebibliography}
\end{document}